\documentclass[12pt]{article}
\usepackage{amsmath}
\usepackage{amsfonts} 
\usepackage{graphicx,psfrag,epsf}
\usepackage{enumerate}
\usepackage{natbib}
\usepackage{url} 
\usepackage{hyperref}
\usepackage{graphicx}
\usepackage{amsthm}
\usepackage{bbm}
\usepackage{soul}
\usepackage{paralist}
\usepackage{mathtools}
\usepackage{adjustbox}
\usepackage{booktabs}
\usepackage{xcolor}
\usepackage{amssymb}
\usepackage{multirow}  
\usepackage{array}     
\usepackage{physics}
\usepackage{tablefootnote}
\usepackage{enumitem}
\usepackage{amssymb}

\usepackage{lipsum}

\usepackage{bbm}
\usepackage{booktabs}
\usepackage{adjustbox}
\usepackage{xcolor}
\usepackage{bbding}

\usepackage{hyperref}
\hypersetup{colorlinks=true, linkcolor=blue, citecolor=blue, filecolor=blue}

\newcommand{\PP}{\mathbb{P}}

\newcommand{\indep}{\perp\!\!\!\perp}
\newcommand{\calM}{\mathcal{M}}
\newcommand{\calX}{\mathcal{X}}

\newcommand{\E}{\mathbb{E}}

\newcommand{\T}{\mathsf{T}}

\newcommand{\Var}{\mathrm{var}}

\DeclareMathOperator*{\esssup}{ess\,sup}
\DeclareMathOperator*{\essinf}{ess\,inf}
\allowdisplaybreaks

\allowdisplaybreaks

\newtheorem{theorem}{Theorem}
\newtheorem{lemma}{Lemma}
\newtheorem{corollary}{Corollary}
\newtheorem{assumption}{Assumption}
\newtheorem{proposition}{Proposition}

\theoremstyle{definition}
\newtheorem{definition}{Definition}

\theoremstyle{remark}
\newtheorem{remark}{Remark}

\usepackage{setspace}
\usepackage{bibspacing}
\newcommand{\blind}{0}

\begin{document}

\def\spacingset#1{\renewcommand{\baselinestretch}%
{#1}\small\normalsize} \spacingset{1}


\if0\blind
{
  \title{\bf 
Sensitivity analysis for causal mediation: bridge score, sharp sensitivity bounds, and calibration

  }
  \author{Yuki Ohnishi$^{1*}$ and Fan Li$^{1 \ddagger}$\\
    $^1$Department of Biostatistics,  
    Yale University School of Public Health, \\
    New Haven, CT, United States\\
    \\
    $^*$yuki.ohnishi@yale.edu, 
    $^\ddagger$fan.f.li@yale.edu}
  \maketitle
} \fi

\if1\blind
{
  \bigskip
  \bigskip
  \bigskip
  \begin{center}
    {\LARGE\bf }
\end{center}
  \medskip
} \fi

\bigskip

\begin{abstract}
Causal mediation analysis decomposes the total treatment effect into a portion operating through a hypothesized mediator and a residual direct portion. Identification of natural direct and indirect effects typically rests on the mediator stage of sequential ignorability, which cannot be empirically verified and requires explicit sensitivity analysis. We formulate the \emph{bridge score}, a mediator-stage balancing score, as a low-dimensional vector formed from the two treatment-specific mediator densities at a common mediator value, and show that it balances baseline covariates for the mediator stage relevant to natural effect identification. Conditional on the bridge score, we derive a sharp pointwise variance envelope on the unidentified mediator-outcome confounding function in terms of latent outcome relevance and residual selection. To make the bound operational for sensitivity analysis, we further introduce a residual budget calibration approach based on local residual outcome variation and record a complementary range bound for support-based restrictions. Finally, we show how the pointwise bound can be operationalized for inference through a scalar functional reduction and a Bayesian g-computation algorithm that combines observed-data posterior uncertainty with user-specified sensitivity uncertainty, rather than treating the unidentified sensitivity corrections as learned from the likelihood.
\end{abstract}

\noindent%
{\it Keywords: causal mediation analysis; sensitivity analysis; bridge score; sharp bound; Bayesian g-computation}
\vfill

\newpage
\spacingset{1.2} 

\section{Introduction}
\label{sec:introduction}
Causal mediation analysis decomposes the average (total) treatment effect into the portion that operates through a hypothesized mediator (referred to as the natural indirect effect) and the portion that works around the mediator (referred to as the natural direct effect). Identification of the natural direct and indirect effects typically rests on the \emph{sequential ignorability} assumption \citep{Imai2010}. That is, treatment assignment must be ignorable, and the mediator must be ignorable given treatment and baseline covariates. The second condition is rarely supported by design, and any unmeasured factor jointly driving the mediator and the outcome leads to its violation. Hence, there remains substantial interest in pursuing a calibrated sensitivity analysis for the no mediator-outcome confounding assumption to strengthen the credibility of causal mediation analysis.

An emerging literature has developed sensitivity analysis for violations of sequential ignorability. \citet{Imai2010} introduced a parametric $\rho$-based parameter tied to the residual correlation between mediator and outcome errors; \citet{Tchetgen2012} developed semiparametric, multiply robust identification with a sensitivity approach through a user-specified, possibly multivariable confounding function; \citet{VanderWeele2010_bias} developed bias formulas and effect decompositions, complemented by the sensitivity analysis \citep{VanderWeele2014_effect_decomposition} and bivariate normal latent confounder analyses \citep{le_Cessie_2016,Lindmark2018}. \citet{Ding2016_sharp}, as well as follow-up work \citep[e.g.,][]{SjolanderWaernbaum_2024}, established risk ratio sharp bounds under interpretable parameters for the strength of confounding; Bayesian treatments include \citet{mccandless2017,Daniels_Linero_Roy_2023}, which emphasize prior factorization into likelihood and sensitivity components.

Despite the aforementioned literature, three challenges remain for credible, nonparametric sensitivity analysis of mediator ignorability on the additive effect measure scale, which is arguably the most commonly used scale to quantify causal effects in scientific research. First, existing parameterizations face a tension between interpretability and localization. Scalar sensitivity parameters are easy to report, but they collapse violations that may vary with the treatment arm, mediator value, and covariate profile; fully localized sensitivity functions are more faithful to the structure of mediator-outcome confounding, but are difficult to elicit, estimate, and calibrate when the covariate is even moderately high-dimensional. 
What has not been fully developed is a general low-dimensional reduction of the sensitivity object that preserves the mediator information relevant to the cross-world natural effect functional without requiring a full covariate specification. A related strand uses treatment-specific mediator densities to construct inverse weights for mediation \citep{TchetgenTchetgen2013_IORW,HongQinYang2018}. However, they do not explicitly develop the mediator-density vector as a balancing score or as the localization scale for additive mediator-outcome sensitivity analysis.
Second, although sharp nonparametric bounds on mediator-outcome confounding under interpretable parameters for the strength of confounding exist on the risk ratio scale and are most natural for binary outcomes \citep{Ding2016_sharp,SjolanderWaernbaum_2024}, no analogous sharp additive envelope is available when the effect measured is defined by a mean difference. Additive sensitivity in the existing literature is handled either by parametric bias formulas that require a specific latent confounding structure \citep{VanderWeele2010_bias,le_Cessie_2016,Lindmark2018} or by Manski-style bounds that use no sensitivity parameters at all and are likely too wide and hence uninformative \citep{Sjolander2009}.
Third, additive sensitivity analysis requires an operational calibration of two unidentified ingredients: residual selection on a latent mediator-outcome confounder and the outcome relevance of that confounder. Directly specifying a latent outcome-relevance parameter on the raw outcome scale can be difficult to interpret across applications, while reducing the analysis to a single robustness scalar hides the distinct roles of residual selection and outcome relevance. A practical calibration therefore needs to keep these two components visible while anchoring the outcome-relevance component to an interpretable local observed-data scale.

In this article, we address these challenges through several contributions. First, we formulate the \emph{bridge score} as a mediator-stage balancing-score analogue of the propensity score. It is a two-dimensional vector formed from both treatment-specific mediator densities at a common mediator value. It places the identified outcome regression and the unidentified mediator-outcome confounding correction on the same low-dimensional scale, while encoding the arm asymmetry in the estimand. In contrast to the density-based methods \citep{TchetgenTchetgen2013_IORW,HongQinYang2018}, which collapse the two arm-specific densities into a single inverse weight for natural effect estimation, the bridge score retains them as a vector and uses them as a balancing score on which the unidentified sensitivity function is defined. Thus, the contribution is not the use of mediator densities alone, but the explicit balancing-score formulation and its use as the organizing scale for localized sensitivity analysis. Second, we derive a sharp pointwise additive variance envelope for the bridge score sensitivity function. At each mediator and bridge score stratum, this envelope is parametrized by two primitive quantities, i.e., the local variance of the latent conditional mean curve that measures outcome relevance of the unmeasured mediator-outcome confounder and a chi-square divergence that measures residual mediator-induced selection in the latent confounder distribution. For calibration, we also use a likelihood-ratio cap and put residual selection on a familiar multiplicative scale. The induced aggregate envelope is sharp over the rectangular class of measurable bridge score sensitivity functions satisfying the pointwise envelope, while narrower latent structural models may impose additional compatibility restrictions. We also record a complementary total-variation (range) envelope, which is useful when outcome support information is available. Third, we provide a residual variance budget calibration that operationalizes the latent mean-variance component by anchoring it to local residual outcome variation, while retaining the residual selection parameter on its likelihood-ratio scale. We then operationalize inference through a Bayesian g-computation implementation that combines observed-data posterior uncertainty in nuisance and calibration quantities with user-specified uncertainty over unidentified sensitivity corrections. To help position our work in the existing literature, Table~\ref{tab:sa_comparison} summarizes the contrast across the dimensions of scale, sensitivity object, calibration route, and target outcome. Detailed positioning of our work relative to existing methods is further discussed in Supplementary Material Section \ref{sec:related_work_summary}.

\begin{table}[h]
\centering
\caption{Connections and differences between the proposed framework and major sensitivity analysis approaches for unmeasured mediator-outcome confounding.}
\label{tab:sa_comparison}
\small
\begin{adjustbox}{width=\textwidth}
\begin{tabular}{llllll}
\toprule
Method & Scale & Sensitivity object & Calibration & Sharp bound & Any outcome type \\
\midrule
\citet{Imai2010} & Additive & $\rho$ residual correlation & \XSolidBrush & \XSolidBrush & \Checkmark \\
\citet{Tchetgen2012} & Additive & Confounding function & \XSolidBrush & \XSolidBrush & \Checkmark \\
\citet{Ding2016_sharp} & risk ratio & $(\mathrm{RR}_{UY},\mathrm{RR}_{AU})$ & \XSolidBrush & \Checkmark & \XSolidBrush  \\
\citet{HongQinYang2018} & Additive & Two sensitivity parameters & \XSolidBrush & \XSolidBrush & \Checkmark \\
\citet{Daniels_Linero_Roy_2023} & Additive & Confounding function & Residual budget & \XSolidBrush & \Checkmark \\
This paper & Additive & $(v_a,\chi_a,\gamma_a)$ on bridge score & Residual budget & \Checkmark & \Checkmark \\
\bottomrule
\end{tabular}
\end{adjustbox}
\end{table}

\section{Data structure, causal estimands and assumptions}
\label{sec:bridge_score_sa}

We consider $n$ independent observations $O_i=(X_i,A_i,M_i,Y_i)$, each drawn from a common superpopulation. Here $A_i\in\{0,1\}$ is the treatment assignment, $X_i\in\calX$ is a vector of pretreatment covariates, $M_i\in\calM$ is a mediator, and $Y_i$ is the outcome. We suppress the unit index when no confusion can arise.
Let $\nu$ be a $\sigma$-finite dominating measure on the support $\calM$ of $M$ such that, for each $a\in\{0,1\}$ and $x\in\calX$, the conditional law of $M$ given $(A=a,X=x)$ has density or mass function $f_a(m\mid x)= dF_{M\mid A=a,X=x}(m) / d\nu$ for $m\in\calM$.
This notation covers both cases of a discrete and a continuous mediator. For $a\in\{0,1\}$ and $m\in\calM$, let $M(a)$ denote the potential mediator under treatment $a$ and $Y(a,m)$ the potential outcome under treatment $a$ and mediator value $m$. Our primary cross-world estimand is the mediated mean under treatment with the mediator set to its natural value under control, $\theta = \E\{Y(1,M(0))\}$.
Let $\delta_a=\E\{Y(a,M(a))\}$ denote the arm-specific observed world means for $a\in\{0,1\}$. 
The natural direct effect (NDE) and indirect effect (NIE) on the additive scale are $\mathrm{NDE}=\theta-\delta_0$ and $\mathrm{NIE}=\delta_1-\theta$,
so that $\mathrm{TE}=\delta_1-\delta_0=\mathrm{NDE}+\mathrm{NIE}$.

We invoke consistency, positivity, and randomization. Treatment randomization is the first stage of sequential ignorability; the second stage, mediator ignorability, is intentionally not assumed and is replaced below by sensitivity functions.

\begin{assumption}[Consistency]
\label{ass:bridge_consistency}
For all $a\in\{0,1\}$ and $m\in\calM$, $A=a \implies M(a)=M$, and $A=a,\ M=m \implies Y(a,m)=Y$.
\end{assumption}

\begin{assumption}[Positivity]
\label{ass:bridge_positivity}
For each $x$ in the support of $X$, $0<\PP(A=1\mid X=x)<1$. For each $a\in\{0,1\}$ and each $(x,m)$ in the support relevant to $\theta$, $f_a(m\mid x)>0$.
\end{assumption}

\begin{assumption}[Randomization]
\label{ass:treatment_randomization}
For all $a,a'\in\{0,1\}$ and $m\in\calM$, $\{Y(a,m),M(a'),X\}\indep A$.
\end{assumption}

Consistency links observed variables to potential outcomes at the realized treatment and mediator. Assumption~\ref{ass:treatment_randomization} is guaranteed by randomized treatment assignment. Standard mediation analysis would add mediator ignorability, $Y(a',m)\indep M(a)\mid A=a,X$; the rest of the paper studies departures from that strong cross-world identification condition. Throughout this paper, we assume treatment randomization and focus exclusively on sensitivity to violations of mediator ignorability. This setting is directly relevant to randomized trials, where treatment ignorability is guaranteed by design and mediator-outcome confounding is often the primary threat to identification. The same logic can be extended to observational studies by introducing additional sensitivity functions for violations of treatment ignorability; we return to this discussion in Section \ref{sec:discussion}.

\section{Bridge score for mediator ignorability}

The cross-world quantity $Y(1,M(0))$ depends simultaneously on the control mediator law and on the treated outcome regression. This motivates a two-component score that retains both mediator laws at the same $(m,x)$.

\begin{definition}[Bridge score]
    \label{def:bridge_score}
    For $m\in\calM$ and $x\in\calX$, define the bridge score vector
    \[
    B(m,x)=\bigl(B_0(m,x),B_1(m,x)\bigr)
    =\bigl(f_0(m\mid x),f_1(m\mid x)\bigr).
    \]
    For a fixed mediator value, write $B(m)=B(m,X)$.
    When the mediator is observed at its realized value, we write $B(M)=B(M,X)$.
\end{definition}
\noindent 
The first component is the control mediator law evaluated at $m$, and the second component is the treated mediator law evaluated at the same $m$. This score is identified from the observed mediator distribution and keeps the two nuisance laws needed for $\theta$ on one scale. 
Throughout our technical development, conditioning on $B(m)$ means conditioning on the fixed-$m$ random variable $B(m,X)$. For estimation and operating our methods, one often computes $B(M,X)$ at the realized mediator; this equals $B(m,X)$ only inside conditional quantities where the realized mediator is $m$, a distinction that is important for continuous mediators. Because $B(m,x)$ is density-valued for continuous $M$, a common one-to-one transformation of the mediator multiplies both density components by the Jacobian at the transformed mediator value and therefore preserves the score fibers at fixed $m$. 
The term ``bridge'' here refers to the role of $B(m,x)$ in coupling the two arm-specific mediator laws at the same mediator value $m$, thereby providing the common balancing scale on which both the identified regression and the unidentified sensitivity correction are indexed. It is unrelated to the ``bridge function'' of proximal causal inference \citep{Miao01102024}, which is a nuisance function solving an integral equation in the proxy variable identification problem.

\begin{lemma}[Balancing property and score-level ignorability]
\label{lem:bridge_balancing}
Fix $a,a'\in\{0,1\}$ and $m\in\calM$. Then $f(x\mid A=a,M=m,B(m)=b)=f(x\mid A=a,B(m)=b)$ for all $(x,b)$ in the support; equivalently, conditional on $(A,B(m))$, the event $M=m$ carries no additional information about $X$. Furthermore, if mediator ignorability $Y(a',m)\indep M \mid A=a,X$ holds, then $Y(a',m)\indep M \mid A=a,B(m)$.
\end{lemma}

The proof is provided in Supplementary Material Section \ref{proof:lem:bridge_balancing}.
The bridge score plays for the mediator a role analogous to the propensity score for treatment assignment \citep{ROSENBAUM1983}. The analogy is about balancing, not about treatment assignment. Once $(A,B(m))$ is fixed, conditioning on the event $M=m$ does not further change the distribution of $X$. Unlike a propensity score, $B(m,x)$ is indexed by the mediator value and is generally vector-valued, because the natural effect target combines the control mediator law and treated outcome regression. It also differs from a principal score \citep[e.g.,][]{Ding2016PScore}. That is, the bridge score is built from identified mediator densities rather than from the unidentified joint law of $(M(0),M(1))$. Next, the primitive unidentified objects are mediator-outcome confounding functions at the score level.

\begin{definition}[bridge score sensitivity function]
\label{def:bridge_sensitivity}
For $a\in\{0,1\}$ and $m\in\calM$, define
\begin{equation}
\label{eq:sensitivity_function}
\Delta_a(m,b)
=
\E\{Y(1,m)\mid A=a,M=m,B(m)=b\}
-
\E\{Y(1,m)\mid A=a,B(m)=b\}.
\end{equation}
\end{definition}

The function $\Delta_a(m,b)$ measures the residual confounding induced by conditioning on the observed mediator value $M=m$ within arm $a$, after conditioning on the bridge score. The subscript $a$ indexes the arm in which this mediator selection is evaluated; the treatment for the potential outcome remains fixed at $1$ because the target $\theta$ uses $Y(1,M(0))$. The analogue indexed by full covariates is
$\Delta_a^\ast(m,x)
=
\E\{Y(1,m)\mid A=a,M=m,X=x\}
-
\E\{Y(1,m)\mid A=a,X=x\}.$
Both functions vanish under exact mediator ignorability.

\begin{proposition}[Projection of the full confounding function onto the bridge score]
\label{prop:bridge_projection}
For fixed $(a,m)$, $\Delta_a(m,b)=\E\!\left[\Delta_a^\ast(m,X)\mid A=a,B(m)=b\right].$
\end{proposition}
The proof is provided in Supplementary Material Section \ref{proof:prop:bridge_projection}.
By Proposition \ref{prop:bridge_projection}, the bridge score sensitivity function is the conditional expectation projection of the full confounding structure indexed by covariates onto the score.
Thus, the bridge score sensitivity function is not an ad hoc simplification; rather, it is a principled, low-dimensional summary of the full confounding structure induced by the bridge score. \citet{Li2011} proposed a sensitivity analysis with a similar principle based on the propensity score for unmeasured treatment-outcome confounding. In our setting, the bridge score sensitivity function plays the corresponding role for mediator ignorability.
To proceed, we define the observed outcome regression on the treated arm $\mu_1(m,b)=\E(Y\mid A=1,M=m,B(m)=b)$, and we obtain the following identification result.

\begin{theorem}[Identification under bridge score sensitivity functions]
\label{thm:bridge_identification}
Suppose Assumptions~\ref{ass:bridge_consistency}--\ref{ass:treatment_randomization} hold. Then
\begin{align}
\theta
=
\E\!\left[
\int_{\calM}
\Bigl\{
\mu_1\bigl(m,B(m)\bigr)-\Delta_1\bigl(m,B(m)\bigr)+\Delta_0\bigl(m,B(m)\bigr)
\Bigr\}
f_0(m\mid X)\,\nu(dm)
\right].
\label{eq:bridge_theta_identified}
\end{align}
\end{theorem}
The proof is provided in Supplementary Material Section \ref{proof:thm:bridge_identification}.
Theorem~\ref{thm:bridge_identification} places the identified component $\mu_1$ and the unidentified correction $(\Delta_0,\Delta_1)$ on the same bridge score scale. The two corrections have different roles. $\Delta_1$ adjusts the treated regression for conditioning on the observed mediator inside the treated arm, whereas $\Delta_0$ adjusts the corresponding transport under the control mediator law. The bridge score is useful for sensitivity analysis because it reduces the covariate dimension while retaining the information of the mediator law at each treatment arm, which determines where the natural effect functional places weight. Thus, calibration can remain local in $(m,b)$ without requiring a sensitivity function indexed by full covariates, making the sensitivity parameters easier to elicit, report, and compare across analyses while still allowing confounding strength to vary over the mediator distribution.

\subsection{Sharp sensitivity bounds based on the bridge score}
\label{sec:sharp_bound}
We next derive sharp bounds for the bridge score sensitivity functions. For bounding only, suppose there exists a latent scalar confounder $U$ such that the residual dependence of $Y(1,m)$ on the mediator within arm $a$ is explained by $(B(m),U)$.

\begin{assumption}[Latent bridge score representation for bounding]
\label{ass:bridge_bound_latent}
For each $a\in\{0,1\}$ and each $m\in\calM$, there exists a latent scalar random variable $U$ such that $Y(1,m)\indep M \mid A=a,\,B(m),\,U$, with $F_{U\mid A=a,M=m,B(m)=b} \ll F_{U\mid A=a,B(m)=b}$. For each $a\in\{0,1\}$ and each $(m,b,u)$ in the relevant support, let 
$\psi_a(u;m,b) = \E\{Y(1,m)\mid A=a,\,B(m)=b,\,U=u\}$
denote the arm-specific regression of the treated potential outcome on the latent confounder within a bridge score stratum.
\end{assumption}
\noindent
Assumption~\ref{ass:bridge_bound_latent} is used only to define a sensitivity bound. It does not alter the observed data model, and it does not require $U$ to be observed.
For $a\in\{0,1\}$ and $(m,b)$, define the latent density ratio
\[
h_a(u;m,b)
=
\frac{
dF_{U\mid A=a,M=m,B(m)=b}
}{
dF_{U\mid A=a,B(m)=b}
}(u).
\]
The bridge score sensitivity function can then be written as a covariance under the reduced law $U\mid A=a,B(m)=b$. Let
\begin{align}
\chi_a(m,b)
&=
\E\!\left[
\{h_a(U;m,b)-1\}^2
\mid A=a,B(m)=b
\right],
\label{eq:chi_local}\\
v_a(m,b)
&=
\Var\{\psi_a(U;m,b)\mid A=a,B(m)=b\}.
\label{eq:latent_mean_variance}
\end{align}
The parameter $\chi_a(m,b)$ is the chi-square divergence between the selected latent law $U\mid A=a,M=m,B(m)=b$ and the reduced latent law $U\mid A=a,B(m)=b$. It measures residual mediator-induced selection on the latent confounder. The parameter $v_a(m,b)$ measures the local second-moment outcome relevance of the latent confounder through variation in the $U$-specific mean of $Y(1,m)$.
For calibration it is often easier to elicit a likelihood-ratio cap than a chi-square divergence. Define
\begin{equation}
\gamma_a(m,b)
=
\esssup_{u\sim F_{U\mid A=a,B(m)=b}}
h_a(u;m,b)
\;\ge\; 1,
\label{eq:gamma_local}
\end{equation}
so that $0\le h_a(u;m,b)\le\gamma_a(m,b)$ under the reduced law, up to null sets. 
The parameter $\gamma_a(m,b)$ measures how much conditioning on the mediator event $M=m$ tilts the distribution of the latent confounder $U$ within arm $a$, after already conditioning on the bridge score $b$. In this sense, it quantifies the magnitude of mediator-induced residual selection on the latent confounder.

\begin{theorem}[Sharp bridge score variance bound]
\label{thm:sharp_bound}
Suppose Assumptions~\ref{ass:bridge_consistency}--\ref{ass:treatment_randomization} and~\ref{ass:bridge_bound_latent} hold. Then, for every $a\in\{0,1\}$ and every $(m,b)$ in the relevant support,
\begin{equation}
\bigl|\Delta_a(m,b)\bigr|
\le
\Xi_{a,\chi}(m,b)
:=
\sqrt{v_a(m,b)\chi_a(m,b)}.
\label{eq:sharp_l2_bound}
\end{equation}
For fixed conditional latent laws $F_{U\mid A=a,M=m,B(m)=b}$ and $F_{U\mid A=a,B(m)=b}$ and fixed $v_a(m,b)$, the bound is sharp among latent outcome regressions with the given second moment and no further range restriction. Moreover,
\begin{equation}
\chi_a(m,b)\le \gamma_a(m,b)-1,
\label{eq:chi_gamma_bound}
\end{equation}
and hence
\begin{equation}
\bigl|\Delta_a(m,b)\bigr|
\le
\Xi_{a,\gamma}(m,b)
:=
\sqrt{v_a(m,b)\{\gamma_a(m,b)-1\}}.
\label{eq:sharp_l2_gamma_bound}
\end{equation}
For finite $\gamma_a(m,b)$, the likelihood-ratio envelope \eqref{eq:sharp_l2_gamma_bound} is sharp over density ratios satisfying $0\le h_a(u;m,b)\le\gamma_a(m,b)$ almost surely under the reduced law and $\E\{h_a(U;m,b)\mid A=a,B(m)=b\}=1$, among latent outcome regressions with fixed $v_a(m,b)$ and no further range restriction.
\end{theorem}
The proof is provided in Supplementary Material Section \ref{proof:thm:sharp_bound}.
Theorem \ref{thm:sharp_bound} presents an analogue of the risk ratio scale bound of \citet{Ding2016_sharp}, adapted to the bridge score formulation and to the more general, mean difference scale.  It decomposes mediator-outcome confounding into an outcome-relevance component, $v_a(m,b)$, and a residual-selection component, $\chi_a(m,b)$ or its likelihood-ratio upper bound $\gamma_a(m,b)-1$. The form of the bound is useful for calibration because $v_a(m,b)$ is a variance of the latent conditional mean curve, and can therefore be scaled against local residual outcome variation. The likelihood-ratio parameter $\gamma_a(m,b)$ remains on the same interpretable selection scale used in sharp confounding bounds of \citet{Ding2016_sharp}. 

Here, sharpness is important because it characterizes the exact admissible range of the primitive sensitivity function $\Delta_a(m,b)$ at each stratum. For fixed $(m,b)$ and fixed values of $(v_a,\chi_a)$, the envelope in Theorem~\ref{thm:sharp_bound} cannot be narrowed without imposing additional restrictions. Specifically, the bound is sharp in the following pointwise sense: for every admissible pair $(v_a,\chi_a)$, there exists a latent confounder $U$ and a joint law of $(Y,M,U)$, consistent with the observed data law at that stratum, such that $|\Delta_a(m,b)|=\sqrt{v_a(m,b)\chi_a(m,b)}$.

For comparison and for optional support restrictions, it is also useful to record the complementary range-based envelope. Define the arm-specific outcome range parameter
\begin{equation}
\eta_a(m,b)
=
\esssup_{u\sim F_{U\mid A=a,B(m)=b}}\psi_a(u;m,b)
-
\essinf_{u\sim F_{U\mid A=a,B(m)=b}}\psi_a(u;m,b)
\;\ge\; 0.
\label{eq:eta_local}
\end{equation}
The parameter $\eta_a(m,b)$ measures the essential spread of the $U$-specific mean of $Y(1,m)$ within the same bridge score stratum. It is an $L^\infty$ outcome-relevance parameter, whereas $v_a(m,b)$ is an $L^2$ outcome-relevance parameter.

\begin{theorem}[Complementary sharp range bound]
\label{thm:range_bound}
Suppose Assumptions~\ref{ass:bridge_consistency}--\ref{ass:treatment_randomization} and~\ref{ass:bridge_bound_latent} hold. Then, for every $a\in\{0,1\}$ and every $(m,b)$ in the relevant support,
\begin{equation}
\bigl|\Delta_a(m,b)\bigr|
\le
\Xi_{a,\mathrm{R}}(m,b)
:=
\eta_a(m,b)\,\frac{\gamma_a(m,b)-1}{\gamma_a(m,b)},
\label{eq:sharp_additive_bound}
\end{equation}
with the convention that $(\gamma_a-1)/\gamma_a=1$ when $\gamma_a=\infty$. The bound is sharp when only the local pair $(\eta_a,\gamma_a)$ and the support admissibility restriction are fixed.
\end{theorem}
The proof is provided in Supplementary Material Section \ref{proof:thm:range_bound}.
Theorem~\ref{thm:range_bound} is not the calibration route used in Section \ref{sec:residual_budget_calibration}, because local residual variance does not identify or bound the essential range $\eta_a(m,b)$. Instead, the range bound supplies a complementary support-based restriction. In particular, if $Y$ has known support $[L_Y,U_Y]$, then $\eta_a(m,b)\le U_Y-L_Y$, and any residual budget envelope can be intersected with $(U_Y-L_Y)\{\gamma_a(m,b)-1\}/\gamma_a(m,b)$.

\subsection{Bridge score and full covariate sensitivity parameters}
\label{sec:bs_for_sharp_bound}
Theorem~\ref{thm:sharp_bound} is sharp conditional on the bridge score parameters. A separate question is whether those parameters should be elicited at the full covariate level or after compression to $B(m)$. The bridge score helps because it balances the observed covariates with respect to the mediator event. Consequently, residual selection parameters defined after conditioning on $B(m)$ can be bounded by their full covariate counterparts within the corresponding score stratum.

\begin{proposition}[bridge score sensitivity parameters and full covariate counterparts]
\label{prop:bridge_tightens}
Fix $a\in\{0,1\}$ and $m\in\calM$. Define the
sensitivity parameters at the covariate level
\begin{align*}
\gamma_a^\ast(m,x) &=
\esssup_{u\sim F_{U\mid A=a,X=x}}
h_a^\ast(u;m,x),\\
h_a^\ast(u;m,x)
&=
\frac{dF(u\mid A{=}a,M{=}m,X{=}x)}{dF(u\mid A{=}a,X{=}x)},\\
\chi_a^\ast(m,x)
&=
\E\!\left[
\{h_a^\ast(U;m,x)-1\}^2
\mid A=a,X=x
\right],\\
v_a^\ast(m,x)
&=
\Var\{\psi_a^\ast(U;m,x)\mid A=a,X=x\},
\end{align*}
where $\psi_a^\ast(u;m,x)=\E\{Y(1,m)\mid A=a,X=x,U=u\}$.
Then, without any additional independence condition,
\begin{align}
\gamma_a(m,b) 
&\le \esssup_{x:\,B(m,x)=b}\gamma_a^\ast(m,x),
\label{eq:gamma_tightens}\\
\chi_a(m,b)
&\le
\E\{\chi_a^\ast(m,X)\mid A=a,B(m)=b\}
\le
\esssup_{x:\,B(m,x)=b}\chi_a^\ast(m,x).
\label{eq:chi_tightens}
\end{align}

Additionally, assuming $X \indep U \mid A=a,B(m)=b$,
we have
\begin{equation}
v_a(m,b) \le \esssup_{x:\,B(m,x)=b}v_a^\ast(m,x).
\label{eq:v_tightens}
\end{equation}
For the complementary range bound, the same condition also gives
\begin{align*}
\eta_a(m,b)
&\le
\esssup_{x:\,B(m,x)=b}\eta_a^\ast(m,x),\\
\text{ where }
\eta_a^\ast(m,x)
&=
\esssup_{u\sim F_{U\mid A=a,X=x}}\psi_a^\ast(u;m,x)
-
\essinf_{u\sim F_{U\mid A=a,X=x}}\psi_a^\ast(u;m,x).
\end{align*}
The essential suprema over $x$ are taken with respect to the conditional law of $X$ given $A=a$ and $B(m)=b$.
\end{proposition}
The proof is provided in Supplementary Material Section \ref{proof:prop:bridge_tightens}.
When the analyst specifies uniform scalar sensitivity parameters, which is the standard mode of use, doing so at the bridge score level avoids carrying worst-case full covariate sensitivity parameters through the analysis. For residual selection, both the likelihood-ratio parameter $\gamma_a$ and the chi-square parameter $\chi_a$ are controlled by their full covariate counterparts within a bridge score stratum. The improvement is largest when there is substantial heterogeneity in $\gamma_a^\ast$, $\chi_a^\ast$, or $v_a^\ast$ across $x$ values that are mapped to the same bridge score, i.e., when the covariates contain information about $M$ that is redundant once the mediator densities $(f_0(m\mid x), f_1(m\mid x))$ are known. 

\begin{remark}[On the assumption $X\indep U\mid A,B(m)$]
\label{sec:x_indep_u}
The inequalities for the outcome-relevance parameters require the additional assumption $X\indep U\mid A=a,B(m)=b$, 
which asserts that observed and unobserved confounders decouple within a bridge score stratum. This assumption is not implied by the balancing property, because the balancing property governs the joint law of $(X,M)$ given $(A,B(m))$ but says nothing about $U$. 
It is plausible when the unmeasured confounder acts through mechanisms orthogonal to the observed covariates, for instance, a latent biological or behavioral trait that is not captured by the baseline $X$, and becomes less plausible when $U$ is strongly correlated with components of $X$, even after adjusting for the assignment and the bridge score. In practice, this asymmetry means that residual selection parameters specified at the bridge score level are controlled by their full covariate counterparts, while $v_a$ and $\eta_a$ require an additional decoupling condition.
\end{remark}

\section{Operational calibration of the sharp sensitivity bound}
\label{sec:calibration_delta}
Theorem~\ref{thm:sharp_bound} separates mediator-outcome confounding into two components: outcome relevance of the latent confounder, measured by the latent mean variance $v_a(m,b)$, and residual selection on the latent confounder, measured by $\chi_a(m,b)$ or by its likelihood-ratio upper bound $\gamma_a(m,b)-1$. These latent quantities are not identified by the observed data alone. We therefore calibrate the sensitivity envelope using observed-data residual variation and treat calibration as a structured sensitivity input rather than an estimator of an unknown identified quantity. In what follows, we keep residual selection on the likelihood-ratio scale through $\gamma_a(m,b)$, and we calibrate $v_a(m,b)$ against local residual outcome variation.

\subsection{Residual variance budget calibration}
\label{sec:residual_budget_calibration}
We calibrate the latent variance $v_a(m,b)$ through the amount of local conditional mean heterogeneity that the latent confounder could induce in the $U$-specific regression of $Y(1,m)$. A residual budget therefore anchors this outcome-relevance component to local residual outcome variation, in the spirit of residual budget analyses \citep[e.g.,][]{Daniels_Linero_Roy_2023}. Let
$$q_{1,MB}(m,b)=\E\{Y\mid A=1,M=m,B(M)=b\}$$ 
and define the local residual scale
\begin{equation}
\sigma_{\mathrm{res}}^2(m,b)
=
\Var\{Y\mid A=1,M=m,B(M)=b\}
=
\E\!\left[
\{Y-q_{1,MB}(m,b)\}^2
\mid A=1,M=m,B(M)=b
\right].
\label{eq:sigma_residual}
\end{equation}
Conditioning on $(M,B(M))$ keeps the calibration on the same local scale as the bridge score sensitivity function $\Delta_a(m,b)$ and the latent outcome regression $\psi_a(u;m,b)$.
The same treated-arm residual scale is used for both arms because both $v_0(m,b)$ and $v_1(m,b)$ describe latent conditional mean heterogeneity for the same treated potential outcome $Y(1,m)$. We impose the residual variance budget restriction
\begin{equation}
v_a(m,b)
\le
\sigma_{\mathrm{res}}^2(m,b)k_a,
\qquad a\in\{0,1\},
\label{eq:variance_budget_calibration}
\end{equation}
for a user-specified $k_a\in[0,1]$. Thus $k_a$ is a $0$-to-$1$ dial for the maximum share of local treated-arm residual variance allocated to latent conditional mean heterogeneity. The boundary $k_a=1$ allows this heterogeneity to use the full local residual variance budget, while $k_a=0$ removes the outcome-relevance component of mediator-outcome confounding. Intermediate values can be read as allowing the latent conditional mean curve to explain at most a fraction $k_a$ of the local residual outcome variation under the calibration.
The intuition is the usual variance-decomposition logic. After conditioning on the observed local information $(M=m,B(M)=b)$, $\sigma_{\mathrm{res}}^2(m,b)$ is the remaining treated-arm outcome variation. If an unmeasured mediator-outcome confounder is outcome-relevant, its contribution should be calibrated against this remaining local variation. Let $\mathcal C_{mb}$ denote the event $\{A=1,M=m,B(M)=b\}$. Were $U$ observed in the treated arm, one could decompose
\[
\Var(Y\mid \mathcal C_{mb})
=
\Var\!\left\{\E(Y\mid \mathcal C_{mb},U)\mid \mathcal C_{mb}\right\}
+
\E\!\left\{\Var(Y\mid \mathcal C_{mb},U)\mid \mathcal C_{mb}\right\}.
\]
The first term is the observed-data analogue of the latent mean-variance component in \eqref{eq:latent_mean_variance}: it is the amount of conditional mean variation attributable to $U$ within the local stratum. Restriction \eqref{eq:variance_budget_calibration} therefore says that the unobserved analogue may use at most a fraction $k_a$ of the local residual outcome variance. This restriction is not an estimate of $v_a(m,b)$, but is a calibration assumption that places the unidentified outcome-relevance component on an interpretable scale, namely the variation left unexplained after conditioning on the mediator and bridge score.

For residual selection, introduce a user-specified likelihood-ratio grid value $g_a\ge 1$ and
\begin{equation}
\gamma_a(m,b)\le g_a,
\qquad a\in\{0,1\}.
\label{eq:gamma_grid_calibration}
\end{equation}
Values such as $g_a \in \{1.1,\,1.25,\,1.5,\,2,\,3\}$ parallel the multiplicative selection ratio scale in sensitivity analysis \citep{Ding2016_sharp,VanderWeele2014_effect_decomposition}. $g_a$ is the largest plausible tilt in the latent confounder density induced by conditioning on $M=m$ after conditioning on the bridge score. The choice $g_a=1$ enforces $\gamma_a(m,b)\equiv 1$ and recovers $\Delta_a(m,b)=0$, the sequential ignorability anchor.
Combining Theorem~\ref{thm:sharp_bound} with \eqref{eq:variance_budget_calibration} and \eqref{eq:gamma_grid_calibration} yields the residual budget envelope
\begin{equation}
\bigl|\Delta_a(m,b)\bigr|
\le
\Xi_{a,\mathrm{RB}}(m,b;k_a,g_a)
:=
\widehat\sigma_{\mathrm{res}}(m,b)\sqrt{k_a(g_a-1)},
\qquad a\in\{0,1\}.
\label{eq:bound_residualbudget}
\end{equation}
Here $(k_a,g_a)$ are the analyst's sensitivity parameters, and $\widehat\sigma_{\mathrm{res}}(m,b)$ is the predicted residual standard deviation from a heteroscedastic working model for \eqref{eq:sigma_residual}, for instance a log-linear residual-variance regression fitted after a working mean model for $q_{1,MB}$. If the outcome support $[L_Y,U_Y]$ is known, the envelope may also be intersected with the support-implied range envelope $(U_Y-L_Y)(g_a-1)/g_a$ from Theorem~\ref{thm:range_bound}. Thus the primary operational bound is the sharp variance bound in Theorem~\ref{thm:sharp_bound} calibrated through \eqref{eq:variance_budget_calibration}, with the range bound serving as an optional admissibility restriction when outcome support information is available.

Table~\ref{tab:calibration_guide} summarizes the primitive sensitivity quantities, user-specified calibration inputs, and observed-data scale quantities used in the operational calibration.

\begin{table}[t]
\centering
\caption{Sensitivity quantities and their operational calibration.}
\label{tab:calibration_guide}
\small
\begin{adjustbox}{width=\textwidth}
\begin{tabular}{@{}>{\raggedright\arraybackslash}p{0.15\textwidth}>{\raggedright\arraybackslash}p{0.28\textwidth}>{\raggedright\arraybackslash}p{0.15\textwidth}>{\raggedright\arraybackslash}p{0.40\textwidth}@{}}
\toprule
Quantity & Role & Status & Operational use \\
\midrule
$v_a(m,b)$ & Latent outcome relevance, measured by variation in the $U$-specific mean of $Y(1,m)$ & Unidentified & Calibrate by $v_a(m,b)\le k_a\sigma_{\mathrm{res}}^2(m,b)$. \\
$\chi_a(m,b)$ & Residual selection on the chi-square divergence scale & Unidentified & Bounded by $\gamma_a(m,b)-1$; primarily a primitive quantity for the sharp variance envelope. \\
$\gamma_a(m,b)$ & Residual selection on a likelihood-ratio scale & Unidentified & Restricted by a user-specified cap $g_a\ge1$. \\
$\eta_a(m,b)$ & Latent conditional-mean range & Unidentified & Optional support-based cap when outcome support is known. \\
$k_a$ & Residual variance budget for latent mean heterogeneity & User-specified & Fraction of local treated-arm residual variance that may be attributed to latent outcome relevance. \\
$g_a$ & Maximum mediator-induced latent density tilt & User-specified & Multiplicative cap on how much conditioning on $M=m$ may tilt the latent confounder law after conditioning on $B(m)$. \\
$\sigma_{\mathrm{res}}(m,b)$ & Local treated-arm residual outcome scale & Modeled from observed data & Anchors the outcome-relevance component to variation left after conditioning on $(M,B(M))$. \\
$\bar{\Xi}_a$ & Aggregated sensitivity envelope & Induced by calibration & Reported admissible scalar bound for $\bar{\Delta}_a$ before prior-weighted posterior summaries. \\
\bottomrule
\end{tabular}
\end{adjustbox}
\end{table}

\subsection{Operationalizing the sensitivity analysis}
\label{sec:operation_SA}
The sharp bridge score variance bound and the residual budget calibration provide a general foundation for sensitivity analysis in causal mediation settings, without restricting the inferential paradigm. As a concrete example, this section gives a Bayesian g-computation implementation that produces posterior sensitivity summaries by combining observed-data posterior uncertainty in nuisance quantities with analyst-specified uncertainty over the calibrated sensitivity corrections \citep{FanLi2023}. We first record the scalar reduction that makes the implementation finite-dimensional.

\begin{corollary}[Scalar functional reduction of the bridge score identification]
\label{cor:scalar_reduction}
Define the sequential ignorability anchor
\begin{equation}
\theta_{\mathrm{SI}}
:=
\E\!\left[
\int_{\calM}
\mu_1\bigl(m,B(m)\bigr)\,
f_0(m\mid X)\,
\nu(dm)
\right],
\label{eq:theta_SI_def}
\end{equation}
and the arm-specific aggregated sensitivity functionals
\begin{equation}
\bar{\Delta}_a
:=
\E\!\left[
\int_{\calM}
\Delta_a\bigl(m,B(m)\bigr)\,
f_0(m\mid X)\,
\nu(dm)
\right],
\qquad
a\in\{0,1\}.
\label{eq:delta_bar}
\end{equation}
Then, under Assumptions~\ref{ass:bridge_consistency}--\ref{ass:treatment_randomization},
\begin{equation}
\theta
=
\theta_{\mathrm{SI}}
+
\bar{\Delta}_0
-
\bar{\Delta}_1,
\label{eq:theta_scalar}
\end{equation}
so $\theta$ depends on the bridge score sensitivity functions through the arm-specific scalars $(\bar{\Delta}_0,\bar{\Delta}_1)$.
If, in addition, Assumption~\ref{ass:bridge_bound_latent} holds and a pointwise envelope $\Xi_a(m,b)$ satisfies $|\Delta_a(m,b)|\le \Xi_a(m,b)$, then
\begin{equation}
\bigl|\bar{\Delta}_a\bigr|
\le
\bar{\Xi}_a
:=
\E\!\left[
\int_{\calM}
\Xi_a\bigl(m,B(m)\bigr)\,
f_0(m\mid X)\,
\nu(dm)
\right],
\label{eq:xi_bar}
\end{equation}
Moreover, \( |\bar{\Delta}_a|\le\bar{\Xi}_a \) is sharp over the rectangular class, where \(\Delta_a(m,b)\) is constrained only by \( |\Delta_a(m,b)|\le\Xi_a(m,b) \) pointwise. Consequently, if \(\bar{\Delta}_0\) and \(\bar{\Delta}_1\) vary independently within their rectangular classes, then
\begin{equation}
\theta\in
\left[
\theta_{\mathrm{SI}}-\bar\Xi_0-\bar\Xi_1,
\theta_{\mathrm{SI}}+\bar\Xi_0+\bar\Xi_1
\right]
\label{eq:theta_aggregate_sharp_interval}
\end{equation}
is sharp over the induced aggregate sensitivity class.
Finally, $\bar{\Xi}_a \le \esssup_{(m,b)}\Xi_a(m,b)$, with equality only when $\Xi_a\{m,B(m)\}$ attains this essential supremum $f_0$-almost surely on the support induced by $f_0(\cdot\mid X)$.
\end{corollary}
The proof is provided in Supplementary Material Section \ref{proof:cor:scalar_reduction}.
Corollary~\ref{cor:scalar_reduction} records a complementary scalar reduction.
The estimand $\theta$ enters the identification formula through the function $\Delta_a(\cdot,\cdot)$ only via the arm-specific aggregated scalar $\bar{\Delta}_a$. This scalar collapse is a property of the $f_0$-weighted integration in \eqref{eq:theta_scalar} and would hold for a covariate version of the identification formula as well. Two practical consequences follow. First, the sensitivity module needs only a finite-dimensional prior on the net correction $\bar{\Delta}=\bar{\Delta}_0-\bar{\Delta}_1$, or equivalently on $(\bar{\Delta}_0,\bar{\Delta}_1)$, rather than a process prior over the function $\Delta_a(\cdot,\cdot)$. 
Second, over the rectangular sensitivity function class, in which $\Delta_a(m,b)$ is allowed to vary across strata subject only to the pointwise envelope $|\Delta_a(m,b)|\le \Xi_a(m,b)$, averaging gives a sharp aggregate envelope rather than merely a conservative supremum bound. The extremal aggregate corrections are attained by taking $\Delta_a(m,b)=\Xi_a(m,b)$ or $\Delta_a(m,b)=-\Xi_a(m,b)$ $f_0$-almost surely on the integration support. 
Moreover, the inequality $\bar{\Xi}_a\le\operatorname*{ess\,sup}_{(m,b)}\Xi_a(m,b)$ is strict unless the pointwise envelope attains its essential supremum $f_0$-almost surely on the integration support. 
The rectangular class implies that if the analyst additionally imposes a common latent structural model, bounded-support compatibility, or dependence restrictions linking $\Delta_0$ and $\Delta_1$, the aggregate interval in \eqref{eq:theta_aggregate_sharp_interval} may be further tightened.

The identification formula in Theorem~\ref{thm:bridge_identification} and Corollary~\ref{cor:scalar_reduction} motivate the following Bayesian g-computation algorithm. Let $\Theta_{\mathrm{obs}}$ denote the collection of parameters indexing the observed data model. At each MCMC iteration $t=1,\dots,T$,
\begin{enumerate}
    \item Draw $\Theta_{\mathrm{obs}}^{(t)}$ from the posterior of the observed data model. This draw induces mediator densities $f_0^{(t)}$ and $f_1^{(t)}$, the treated outcome regression $\mu_1^{(t)}$, and, when natural effects are reported, observed world mean draws $\delta_0^{(t)}$ and $\delta_1^{(t)}$.
    \item For each unit $i=1,\ldots,n$, sample $L$ mediator draws $M_{i0}^{(t,\ell)}$ from $f_0^{(t)}(\cdot\mid X_i)$, $\ell=1,\ldots,L$, and compute $B_{i0}^{(t,\ell)} = \left\{f_0^{(t)}\bigl(M_{i0}^{(t,\ell)}\mid X_i\bigr), f_1^{(t)}\bigl(M_{i0}^{(t,\ell)}\mid X_i\bigr)\right\}$.
    \item For each $(i,\ell)$ and arm $a$, compute the pointwise calibrated envelope $\Xi_{a,\mathrm{RB}}^{(t)}\bigl(M_{i0}^{(t,\ell)},B_{i0}^{(t,\ell)};k_a,g_a\bigr)$ using residual budget calibration \eqref{eq:bound_residualbudget}.
    \item Approximate the scalar envelope in \eqref{eq:xi_bar} by
    \begin{align*}
        \bar{\Xi}_a^{(t)}
        \approx
        \frac{1}{nL}\sum_{i=1}^{n}\sum_{\ell=1}^{L}
        \Xi_{a,\mathrm{RB}}^{(t)}\bigl(M_{i0}^{(t,\ell)},B_{i0}^{(t,\ell)};k_a,g_a\bigr),
    \end{align*}
    where the outer expectation over $X$ is approximated by the empirical mean \citep{FanLi2023}.
    \item Draw the scalar sensitivity function $\bar{\Delta}_a^{(t)}$ from a user-specified sensitivity prior supported on $[-\bar{\Xi}_{a}^{(t)},\bar{\Xi}_{a}^{(t)}]$, such as a uniform prior or a shifted-scaled beta prior. 
    \item By Corollary~\ref{cor:scalar_reduction}, the posterior sensitivity draw of the mediated mean is
    \begin{equation*}
        \theta^{(t)}
        =
        \theta_{\mathrm{SI}}^{(t)}
        +\bar{\Delta}_0^{(t)} - \bar{\Delta}_1^{(t)},
    \end{equation*}
    where $\theta_{\mathrm{SI}}^{(t)} = (nL)^{-1}\sum_{i=1}^n\sum_{\ell=1}^L \mu_1^{(t)}\bigl(M_{i0}^{(t,\ell)},B_{i0}^{(t,\ell)}\bigr)$ is the Monte Carlo sequential ignorability anchor from \eqref{eq:theta_SI_def}. The natural effect draws are $\mathrm{NIE}^{(t)}=\delta_1^{(t)}-\theta^{(t)}$ and $\mathrm{NDE}^{(t)}=\theta^{(t)}-\delta_0^{(t)}$.
\end{enumerate}

The draws produced by this algorithm are posterior sensitivity summaries rather than ordinary posterior draws from an identified likelihood. The observed-data posterior updates the nuisance components that determine $\theta_{\mathrm{SI}}$ and $\bar\Xi_a$, whereas the sensitivity prior for $(\bar\Delta_0,\bar\Delta_1)$ is not updated by the likelihood because these corrections are unidentified without additional assumptions. We therefore recommend reporting the admissible envelope implied by $\bar\Xi_a$ before reporting prior-weighted posterior summaries.

\section{Illustration with immigration framing data}
\label{sec:framing_illustration}

We illustrate the calibrated pointwise bridge score variance envelope in the context of causal mediation analysis, where the risk ratio scale is not the natural reporting scale, and show how inference can be carried out directly from the derived sensitivity bounds. The data are the Brader, Valentino and Suhay immigration framing experiment \citep{BraderValentinoSuhay2008}, distributed as the \texttt{framing} data set in the \texttt{mediation} R package, with 265 respondents. We take $A=\texttt{treat}$, the product of the negative tone and Latino immigrant framing indicators used to mark the experimentally salient treatment cell, as the treatment; $M=\texttt{emo}$, a negative feeling score measured during the experiment, as the mediator; $Y=\texttt{p\_harm}$, perceived harm from increased immigration on a 2-8 scale, as the outcome; and age, education, gender and income as baseline covariates $X$. In this application the natural indirect effect, $\mathrm{NIE}=E\{Y(1,M(1))-Y(1,M(0))\}$, is the part of the framing effect on perceived immigration harm operating through treatment-induced changes in emotional response, whereas the natural direct effect, $\mathrm{NDE}=E\{Y(1,M(0))-Y(0,M(0))\}$, is the remaining effect of the frame with the emotional response distribution held at its control value.

For inference, we use the bridge score specification with linear working models. The mediator model is $M\mid A,X\sim N(x_M^\T\beta_M,\sigma_M^2)$ with regressors $(1,A,X)$. Given the fitted mediator laws, let $\ell_a(m,x)=\log f_a(m\mid x)$ for $a=0,1$. The outcome mean and variance models are $Y\mid A,M,B(M) \sim N\{x_Y^\T\beta_Y,\exp(z_Y^\T\alpha_Y)\}$, $x_Y=\{1,M,A,\ell_0,\ell_1,M\ell_0,M\ell_1,\ell_0\ell_1\}$, $z_Y=(1,M,A,\ell_0,\ell_1)$.
The variance design is deliberately simpler than the mean design. For each simulation, we first draw $(\beta_M,\sigma_M^2)$ from the conjugate mediator posterior and then, conditional on the mediator-law draw and resulting bridge score design, draw $(\beta_Y,\alpha_Y)$ from the heteroscedastic outcome posterior using a weighted Gaussian update for $\beta_Y$ and a random-walk Metropolis update for $\alpha_Y$. The fitted mediator laws produce the log bridge scores at the observed and counterfactual mediator values, and the outcome mean model gives draws of $\delta_0$, $\delta_1$, and $\theta_{\mathrm{SI}}$. We then perform the Bayesian g-computation approach detailed in Section \ref{sec:operation_SA} to obtain the draw of $\theta$.

For the sensitivity layer, we use the residual budget calibration in Section~\ref{sec:residual_budget_calibration}. The calibration imposes $v_a(m,b)\le k_a\sigma_{\mathrm{res}}^2(m,b)$ and varies the residual selection grid value $g_a\ge1$, giving the pointwise envelope $\Xi_{a,\mathrm{RB}}(m,b;k_a,g_a)=\widehat\sigma_{\mathrm{res}}(m,b)\sqrt{k_a(g_a-1)}$. In this heteroscedastic implementation, $\widehat\sigma_{\mathrm{res}}^{(t)}(m,b)=\exp\{z_Y(m,b)^\T\alpha_Y^{(t)}/2\}$ is evaluated at each counterfactual $M(0)$ Monte Carlo draw and then averaged, so the reported aggregate envelope preserves the pointwise residual-scale calibration. Because $Y=\texttt{p\_harm}$ is bounded on the 2--8 scale, we also apply the support-implied range envelope $6(g_a-1)/g_a$ from Theorem~\ref{thm:range_bound} pointwise before averaging; this cap is active in some posterior draws and sparse bridge score regions. Thus $k_a$ controls the share of local residual outcome variation allocated to latent conditional mean heterogeneity, and $g_a$ controls the largest density tilt of the latent confounder induced by conditioning on the mediator after conditioning on the bridge score. For display, we set $k_0=k_1=k$ and $g_0=g_1=g$.

\begin{figure}[t]
\centering
\includegraphics[width=0.82\textwidth]{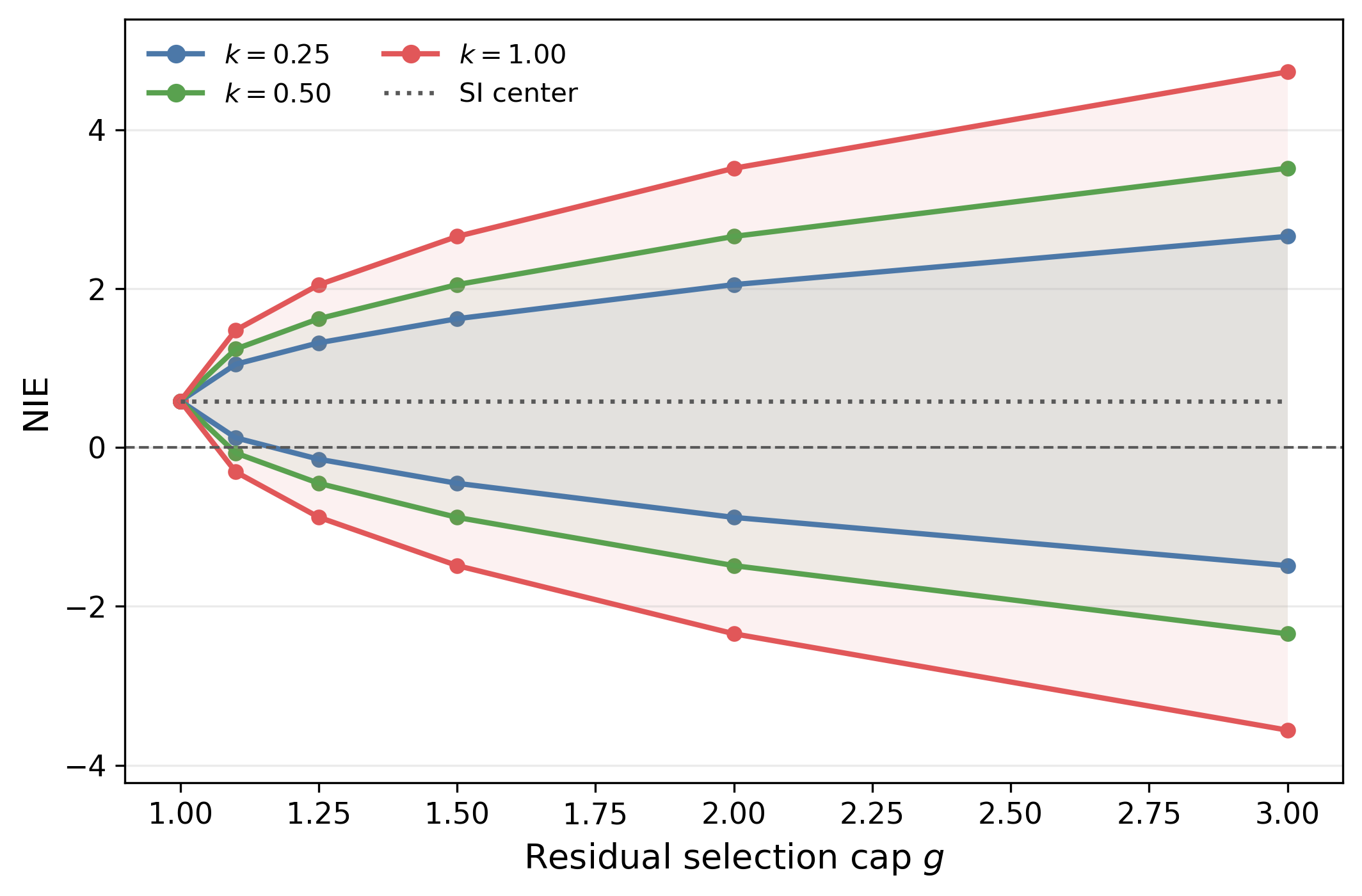}
\caption{Residual budget sensitivity analysis for the immigration-framing illustration. Curves are directional lower and upper NIE envelopes over the residual selection grid value $g$, overlaying $k\in\{0.25,0.50,1.00\}$. Shaded regions mark admissible intervals; the dotted line marks the sequential ignorability center.}
\label{fig:framing_residual_envelope}
\end{figure}

Figure~\ref{fig:framing_residual_envelope} presents the residual budget sensitivity sweep. The sequential-ignorability posterior mean NIE is about $0.59$, and the posterior mean of the averaged pointwise residual scale used in the calibration is about $1.47$ with a 95\% posterior interval of $[1.19,1.78]$. Increasing $k$ expands the allowed local latent mean-variance budget and increasing $g$ increases the permitted residual selection, so the NIE envelope widens monotonically in both dimensions. The envelope collapses to the sequential ignorability center at $g=1$, and wider residual budgets allow the lower endpoint to include zero at smaller values of $g$. In the displayed grid, the near-null value $g=1.1$ leaves the posterior mean lower envelope positive for $k=0.25$, but below zero for $k=0.50$ and $k=1.00$; by $g=1.25$, the lower envelope is below zero for all three values of $k$. These patterns are mechanically implied by the residual budget envelope, but they are useful diagnostics: they show how the robustness conclusion changes as the analyst allocates more local residual outcome variation and residual mediator-induced selection to possible unmeasured mediator-outcome confounding.

\section{Discussion}\label{sec:discussion}

This paper develops a sensitivity analysis framework for mediator-outcome ignorability that operates on the additive mean-difference scale. The central object is the bridge score, a two-dimensional vector formed from both treatment-specific mediator densities at a common mediator value, which serves as a balancing score for the mediator and as a low-dimensional localization scale for the unidentified confounding correction. Building on this score, we derive a sharp pointwise variance envelope for the bridge score sensitivity function, parametrized by latent mean variance and residual selection. We also record a complementary total-variation (range) envelope for support-based restrictions. 
Residual budget calibration translates the latent mean-variance component of the sharp-bound construction into a local residual variance budget, while preserving the same bridge score localization scale and likelihood-ratio residual selection parameter. 
We then illustrate the operationalization of the framework through a Bayesian g-computation sensitivity prior algorithm that propagates observed-data nuisance uncertainty and user-specified sensitivity uncertainty, without claiming that the unidentified corrections are learned from the likelihood.

These contributions are modular. The bridge score can be studied as a mediator-stage balancing score and used as a dimension reduction tool for matching, stratification and weighting, even apart from the particular sensitivity bound developed here. The sharp variance envelope is a pointwise result for additive mediator-outcome confounding and can be combined with other ways of specifying or estimating the localization scale. The residual budget calibration is likewise an operational layer, not a feature of the Bayesian algorithm itself; it can be paired with plug-in, frequentist or Bayesian summaries whenever a local residual variance scale is available. Thus, although the paper combines bridge-score localization, sharp bounding and residual-scale calibration into one workflow, each component has independent methodological content and can be adapted separately.
Additionally, we formulate the framework under the treatment randomization assumption, $\{Y(a,m),M(a'),X\}\indep A$, to emphasize the role of the bridge score both as a balancing device for mediator ignorability and as the scale on which mediator-outcome sensitivity is indexed. The same logic could be adapted to settings under conditional treatment exchangeability, $\{Y(a,m),M(a')\}\indep A \mid X$, by augmenting the localization variables with a treatment assignment balancing score, for example $(e(X),f_0(m\mid X),f_1(m\mid X))$ with $e(X)=\PP(A=1\mid X)$. We leave this extension implicit to keep the present framework centered on mediator-outcome sensitivity.

The framework is tailored to the cross-world natural effect estimand, where the asymmetry between the two treatment arms is the defining structural feature, but the same logic opens several natural directions for extension. First, causal pathways can operate through multiple mediators simultaneously, whether structured into a causal sequence or acting in parallel without a prespecified ordering \citep{vanderweele2013_multiple}. Extending the bridge score to this setting raises questions about whether a single joint score can preserve the arm asymmetry that the natural effect functional depends on, or whether the score must be defined separately for each mediator. Second, we have primarily focused on independent observations, but in other applied settings (e.g., cluster-randomized experiments) where causal mediation analysis is of interest, outcomes and mediators are measured for units within clusters, and the bridge score sensitivity framework would need to account for within-cluster dependence in both the mediator law and the confounding correction \citep{Cheng2026}. Furthermore, \citet{ohnishi2025} have recently developed a Bayesian nonparametric approach for studying the natural indirect effect and the spillover mediation effects in cluster-randomized experiments with multiple unstructured mediators, and it would be of interest to generalize the current framework to similar cluster-correlated data settings. 

\section*{Acknowledgement}
Research in this article was supported by the Patient-Centered Outcomes Research Institute\textsuperscript{\textregistered} (PCORI\textsuperscript{\textregistered} Award ME-2023C1-31350). The statements presented in this article are solely the responsibility of the authors and do not necessarily represent the views of PCORI\textsuperscript{\textregistered}, its Board of Governors or Methodology Committee.

\bibliographystyle{plainnat}
\bibliography{literature}

\newpage
\appendix
\renewcommand{\theequation}{A\arabic{equation}}
\setcounter{equation}{0}
\setcounter{page}{1}

\section{Positioning relative to existing sensitivity analyses}
\label{sec:related_work_summary}

Among the prior work surveyed above, two families are widely used in applied mediation sensitivity analysis and warrant explicit contrast with our framework. One is the parametric analysis based on $\rho$ proposed by \citet{Imai2010}; the other is the risk ratio sharp bound family of \citet{Ding2016_sharp,SjolanderWaernbaum_2024}, together with its scalar E-value reduction \citep{smith2019mediational}. We position the bridge score framework as a complementary tool rather than a replacement. 
The $\rho$-based analysis of \citet{Imai2010} operates on the additive scale but requires fully specified parametric models for both the outcome and the mediator, a restrictive assumption that may fail under model misspecification and that constrains the form of mediator-outcome confounding the analyst can encode. Our bound is nonparametric. It imposes no functional-form restriction on the outcome regression or the mediator law, and the sensitivity parameters separate latent outcome relevance from residual selection rather than using residual-correlation parameters from a working parametric model. The risk ratio bound of \citet{smith2019mediational} is the natural sharp envelope on the multiplicative scale, especially for binary outcomes, and we supply an additive variance-bound counterpart for any outcome type, equipped with bridge score dimension reduction and a residual-scale calibration route. Given such considerations, our framework has the following three differentiating features:
\vspace{-0.05in}
\begin{enumerate}
    \item \emph{Scale.} The sharp bound of \citet{Ding2016_sharp} operates on the risk ratio scale and is most natural for binary outcomes, where the risk ratio is the canonical effect measure; applying it to a general outcome requires either dichotomization or a log-linear approximation. Section~\ref{sec:sharp_bound} develops an additive variance bound on the mean difference scale, so for any outcome type, the analyst can conduct sensitivity analysis without leaving the scale of the natural effect.
    \item \emph{Dimension reduction.} Existing sensitivity analyses do not use the mediator-density bridge score as the localization scale; in common implementations their sensitivity parameters are global or attached to the full covariate/model scale. Proposition~\ref{prop:bridge_tightens} shows that, for the additive bound, calibration on the bridge score can avoid carrying worst-case full covariate sensitivity parameters through the analysis. The likelihood-ratio residual selection parameter at the bridge score level is no larger than its worst-case full covariate counterpart, and the chi-square residual selection parameter is no larger than the corresponding within-stratum average, hence no larger than the within-stratum worst case. The analogous latent mean-variance comparison requires the additional condition stated in the proposition. The bridge score reduction is therefore not only conceptually cleaner but remains informative quantitatively, and it is what makes the additive bound practical when the baseline covariate vector is moderately high-dimensional.
    \item \emph{Residual-scale calibration.} The residual budget calibration in Section~\ref{sec:residual_budget_calibration} translates latent outcome relevance into a variance budget relative to local treated-arm residual variation after conditioning on the mediator and bridge score, while leaving the residual selection parameter on a direct likelihood-ratio scale. This keeps the two components of the additive sensitivity analysis visible and gives the analyst a transparent way to vary outcome relevance and mediator-induced selection separately. Among sensitivity analyses formulated on an additive scale, neither \citet{Imai2010} nor the weighting-based approach of \citet{HongQinYang2018} provides this residual budget calibration for a bridge score envelope.
\end{enumerate}

\section{Proofs}

\subsection{Proof of Lemma~\ref{lem:bridge_balancing}}
\label{proof:lem:bridge_balancing}
\begin{proof}
    Throughout this proof, conditional densities and conditional laws are understood as regular versions. We write $f_{X\mid A=a,M,B(m)}(x\mid m,b)$ and $f_{X\mid A=a,B(m)}(x\mid b)$ for conditional Radon--Nikod\'ym derivatives of $X$ given the indicated variables, and interpret conditioning on $M=m$ through the conditional density $f_a(m\mid x)$ when $M$ is continuous. By Bayes' rule, for almost every $b$,
    \[
    f_{X\mid A=a,M,B(m)}(x\mid m,b)
    \propto
    f_a(m\mid x)\, f_{X\mid A=a,B(m)}(x\mid b),
    \]
    where the proportionality is in $x$. Because $B(m,x)=\bigl(f_0(m\mid x),f_1(m\mid x)\bigr)$, $f_a(m\mid x)=b_a$ for all $x$ in the fiber $\{x:B(m,x)=b\}$, up to the usual null-set convention. This factor is constant in $x$ on the fiber, and normalization gives
    \[
    F_{X\mid A=a,M=m,B(m)=b}=F_{X\mid A=a,B(m)=b}.
    \]
    This proves the balancing property.
    
    To prove pointwise mediator-event ignorability at the distribution level, let $\varphi$ be any bounded measurable function. By iterated expectation,
    \begin{align*}
    &\E\{\varphi(Y(a',m))\mid A=a,M=m,B(m)=b\}\\
    &\qquad=
    \E\left[\E\{\varphi(Y(a',m))\mid A=a,M=m,X\}\mid A=a,M=m,B(m)=b\right].
    \end{align*}
    Under distributional mediator ignorability given $(A,X)$, the inner conditional expectation equals $\E\{\varphi(Y(a',m))\mid A=a,X\}$. Applying the balancing property to this function of $X$ yields
    \[
    \E\{\varphi(Y(a',m))\mid A=a,M=m,B(m)=b\}
    =
    \E\{\varphi(Y(a',m))\mid A=a,B(m)=b\}.
    \]
    Since this equality holds for every bounded measurable $\varphi$, the conditional law of $Y(a',m)$ given $(A=a,M=m,B(m)=b)$ equals the conditional law given $(A=a,B(m)=b)$, proving the fixed-$m$ equality of conditional laws stated in the lemma.
\end{proof}

\subsection{Proof of Proposition \ref{prop:bridge_projection}}
\label{proof:prop:bridge_projection}
\begin{proof}
    Fix $(a,m,b)$. By iterated expectation,
    \begin{align*}
        &\E\{Y(1,m)\mid A=a,M=m,B(m)=b\}\\
        &\qquad=
        \E\!\left[\E\{Y(1,m)\mid A=a,M=m,X\}\mid A=a,M=m,B(m)=b\right]\\
        &\qquad=
        \E\!\left[\E\{Y(1,m)\mid A=a,X\}+\Delta_a^\ast(m,X)\mid A=a,M=m,B(m)=b\right].
    \end{align*}
    By Lemma~\ref{lem:bridge_balancing},
    \[
    \E\!\left[\E\{Y(1,m)\mid A=a,X\}\mid A=a,M=m,B(m)=b\right]
    =
    \E\{Y(1,m)\mid A=a,B(m)=b\}.
    \]
    Subtracting the latter term gives
    \[
    \Delta_a(m,b)
    =
    \E\!\left[\Delta_a^\ast(m,X)\mid A=a,M=m,B(m)=b\right].
    \]
    A second application of Lemma~\ref{lem:bridge_balancing} yields the stated formula.
\end{proof}

\subsection{Proof of Theorem \ref{thm:bridge_identification}}
\label{proof:thm:bridge_identification}
\begin{proof}
    The proof proceeds in two steps. Step~1 identifies $\E\{Y(1,m)\mid M(0)=m,X\}$ at the full covariate level by combining consistency, treatment randomization, and the sensitivity function at the covariate level $\Delta_a^\ast$. Step~2 projects the covariate level identification formula onto the bridge score using the balancing property of Lemma~\ref{lem:bridge_balancing} and Proposition~\ref{prop:bridge_projection}, and applies Fubini's theorem to obtain \eqref{eq:bridge_theta_identified}.
    
    By iterated expectation,
    \[
    \theta
    =
    \E\!\left[
    \int_{\calM}
    \E\{Y(1,m)\mid M(0)=m,X\}
    \, dF_{M(0)\mid X}(m\mid X)
    \right].
    \]
    By Assumption~\ref{ass:treatment_randomization} and Assumption~\ref{ass:bridge_consistency}, the conditional law of $M(0)$ given $X$ equals the observed control-arm law, so we have
    \[
    dF_{M(0)\mid X}(m\mid X)=f_0(m\mid X)\,\nu(dm).
    \]
    
    \medskip
    \noindent\textit{Step 1: covariate level identification.}
    Fix $m\in\calM$ and $x\in\calX$.
    By Assumption~\ref{ass:treatment_randomization}, treatment is randomized with respect to $(Y(1,m),M(0),X)$, so we have
    \[
    \E\{Y(1,m)\mid M(0)=m,X=x\}
    =
    \E\{Y(1,m)\mid A=0,M(0)=m,X=x\}.
    \]
    By Assumption~\ref{ass:bridge_consistency}, when $A=0$ we have $M(0)=M$, so
    \[
    \E\{Y(1,m)\mid M(0)=m,X=x\}
    =
    \E\{Y(1,m)\mid A=0,M=m,X=x\}.
    \]
    By definition of $\Delta_0^\ast(m,x)$,
    \[
    \E\{Y(1,m)\mid A=0,M=m,X=x\}
    =
    \E\{Y(1,m)\mid A=0,X=x\}+\Delta_0^\ast(m,x).
    \]
    By Assumption~\ref{ass:treatment_randomization},
    \[
    \E\{Y(1,m)\mid A=0,X=x\}
    =
    \E\{Y(1,m)\mid A=1,X=x\}.
    \]
    By Assumption~\ref{ass:bridge_consistency},
    $\E(Y\mid A{=}1,M{=}m,X{=}x)=\E\{Y(1,m)\mid A{=}1,M{=}m,X{=}x\}$.
    Subtracting $\Delta_1^\ast(m,x)$,
    \[
    \E\{Y(1,m)\mid A=1,X=x\}
    =
    \E(Y\mid A=1,M=m,X=x)-\Delta_1^\ast(m,x).
    \]
    Combining the preceding displays,
    \begin{equation}
        \label{eq:cov_level_id}
        \E\{Y(1,m)\mid M(0)=m,X=x\}
        =
        \E(Y\mid A{=}1,M{=}m,X{=}x)
        -\Delta_1^\ast(m,x)+\Delta_0^\ast(m,x),
    \end{equation}
    and therefore
    \[
    \theta
    =
    \E\!\left[
    \int_{\calM}
    \bigl\{
    \E(Y\mid A{=}1,M{=}m,X)-\Delta_1^\ast(m,X)+\Delta_0^\ast(m,X)
    \bigr\}
    \,f_0(m\mid X)\,\nu(dm)
    \right].
    \]
    
    \medskip
    \noindent\textit{Step 2: projection onto the bridge score.}
    Since $f_0(m\mid X)=B_0(m,X)$ is $\sigma\bigl(B(m)\bigr)$-measurable, the tower property gives, for any function $g(m,X)$ and each fixed $m$,
    \[
    \E\bigl[g(m,X)\,f_0(m\mid X)\bigr]
    =
    \E\bigl[f_0(m\mid X)\cdot\E\{g(m,X)\mid B(m)\}\bigr].
    \]
    We claim that, for any integrable function $g$, $\E\{g(m,X)\mid B(m)\}$ equals the corresponding bridge score-level quantity for each of the three terms in \eqref{eq:cov_level_id}.
    
    For the outcome regression, note that $h(X):=\E(Y\mid A{=}1,M{=}m,X)$ is a function of $X$ alone. By Assumption~\ref{ass:treatment_randomization}, $X\indep A$, so
    $\E\{h(X)\mid B(m)\}=\E\{h(X)\mid A{=}1,B(m)\}$.
    By Lemma~\ref{lem:bridge_balancing},
    $\E\{h(X)\mid A{=}1,B(m)\}=\E\{h(X)\mid A{=}1,M{=}m,B(m)\}$.
    By iterated expectations, the right-hand side equals $\E(Y\mid A{=}1,M{=}m,B(m))=\mu_1\bigl(m,B(m)\bigr)$.
    
    For the sensitivity functions, Proposition~\ref{prop:bridge_projection} and Assumption~\ref{ass:treatment_randomization} give
    $\E\{\Delta_a^\ast(m,X)\mid B(m)\}
    =\E\{\Delta_a^\ast(m,X)\mid A{=}a,B(m)\}
    =\Delta_a\bigl(m,B(m)\bigr)$ for $a\in\{0,1\}$.
    
    Substituting and applying Fubini's theorem,
    \[
    \theta
    =
    \E\!\left[
    \int_{\calM}
    \Bigl\{
    \mu_1\bigl(m,B(m)\bigr)-\Delta_1\bigl(m,B(m)\bigr)+\Delta_0\bigl(m,B(m)\bigr)
    \Bigr\}
    f_0(m\mid X)\,\nu(dm)
    \right],
    \]
    which is \eqref{eq:bridge_theta_identified}.
\end{proof}

\subsection{Proof of Theorem \ref{thm:sharp_bound}}
\label{proof:thm:sharp_bound}
\begin{proof}
    Fix $a\in\{0,1\}$ and $(m,b)$ throughout. Write
    \[
    F_1(du)=F_{U\mid A=a,M=m,B(m)=b}(du),
    \qquad
    F_0(du)=F_{U\mid A=a,B(m)=b}(du),
    \]
    let $h=dF_1/dF_0$, and let $\E_0$ denote expectation under $F_0$.
    By Assumption~\ref{ass:bridge_bound_latent} and iterated expectation,
    \[
    \Delta_a(m,b)
    =
    \int \psi_a(u;m,b)\{F_1(du)-F_0(du)\}
    =
    \E_0[\{h(U)-1\}\psi_a(U;m,b)].
    \]
    Since $h$ is the Radon--Nikod\'ym derivative of a probability law with respect to $F_0$,
    $\E_0\{h(U)\}=1$, and hence $\E_0\{h(U)-1\}=0$. Therefore
    \[
    \Delta_a(m,b)
    =
    \E_0\!\left[
    \{h(U)-1\}
    \{\psi_a(U;m,b)-\E_0\psi_a(U;m,b)\}
    \right].
    \]
    Applying Cauchy--Schwarz gives
    \[
    |\Delta_a(m,b)|
    \le
    \sqrt{
    \E_0\!\left[
    \{\psi_a(U;m,b)-\E_0\psi_a(U;m,b)\}^2
    \right]}
    \sqrt{\E_0[\{h(U)-1\}^2]}
    =
    \sqrt{v_a(m,b)\chi_a(m,b)}.
    \]

    This inequality is sharp for fixed $(F_0,F_1)$ and fixed $v_a(m,b)$. If $\chi_a(m,b)=0$, then $h=1$ $F_0$-almost surely and both sides are zero. If $\chi_a(m,b)>0$, set
    \[
    \psi_a(u;m,b)=\mu+c\{h(u)-1\},
    \qquad
    c=\sqrt{\frac{v_a(m,b)}{\chi_a(m,b)}},
    \]
    for an arbitrary constant $\mu$. Then
    $\Var_0\{\psi_a(U;m,b)\}=c^2\chi_a(m,b)=v_a(m,b)$ and
    \[
    \Delta_a(m,b)
    =
    c\,\E_0[\{h(U)-1\}^2]
    =
    \sqrt{v_a(m,b)\chi_a(m,b)}.
    \]
    Taking $c$ with the opposite sign attains the lower endpoint.

    It remains to relate $\chi_a$ to $\gamma_a$. Because $0\le h(u)\le \gamma_a(m,b)$ $F_0$-almost surely,
    \[
    \chi_a(m,b)
    =
    \E_0[h^2(U)]-1
    \le
    \gamma_a(m,b)\E_0[h(U)]-1
    =
    \gamma_a(m,b)-1.
    \]
    Combining this with the first bound gives \eqref{eq:sharp_l2_gamma_bound}.

    For finite $\gamma_a(m,b)>1$, the likelihood-ratio envelope is sharp over density ratios satisfying $0\le h\le \gamma_a(m,b)$ and $\E_0h=1$. Let $F_0$ be the uniform law on $[0,1]$, set $C=[0,1/\gamma_a(m,b)]$, and define $h(u)=\gamma_a(m,b)\mathbf{1}\{u\in C\}$. Then $h$ is a probability density ratio, $\E_0[h]=1$, and $\chi_a(m,b)=\gamma_a(m,b)-1$. Choosing $\psi_a(u;m,b)=\mu+c\{h(u)-1\}$ with
    $c=\sqrt{v_a(m,b)/\{\gamma_a(m,b)-1\}}$ attains
    $\sqrt{v_a(m,b)\{\gamma_a(m,b)-1\}}$. The case $\gamma_a(m,b)=1$ is degenerate with $h=1$ and $\Delta_a(m,b)=0$.
\end{proof}

\subsection{Proof of Theorem \ref{thm:range_bound}}
\label{proof:thm:range_bound}
The proof has four steps. Step~1 represents $\Delta_a(m,b)$ as a signed integral of the latent outcome regression $\psi_a(u;m,b)$ against the selected and reduced laws of $U$. Step~2 rescales $\psi_a$ to a bounded function in $[0,1]$ using the essential outcome range parameter $\eta_a(m,b)$. Step~3 applies a total-variation inequality under the likelihood-ratio bound $dF_1/dF_0\le \gamma_a(m,b)$ almost surely, with the convention $r(\infty)=1$. Step~4 gives a local construction attaining both endpoints when only $(\eta_a,\gamma_a)$ and the support admissibility restriction are fixed and $\gamma_a<\infty$.

\begin{proof}
    Fix $a\in\{0,1\}$ and $(m,b)$ throughout. Define
    \[
    F_1(du)=F_{U\mid A=a,M=m,B(m)=b}(du),
    \qquad
    F_0(du)=F_{U\mid A=a,B(m)=b}(du),
    \]
    and let $h=dF_1/dF_0$. By definition, $0\le h\le \gamma_a(m,b)$ $F_0$-almost surely and $\int h\,dF_0=1$.
    
    \medskip
    \noindent\textit{Step 1: integral representation.}
    By Assumption~\ref{ass:bridge_bound_latent} and iterated expectation,
    \begin{align*}
        \E\{Y(1,m)\mid A=a,M=m,B(m)=b\}
        &=
        \int \psi_a(u;m,b)\,F_1(du),\\
        \E\{Y(1,m)\mid A=a,B(m)=b\}
        &=
        \int \psi_a(u;m,b)\,F_0(du).
    \end{align*}
    Therefore
    \begin{equation}
        \Delta_a(m,b)=\int \psi_a(u;m,b)\{F_1(du)-F_0(du)\}.
        \label{eq:delta_integral_appendix}
    \end{equation}
    
    \medskip
    \noindent\textit{Step 2: reduction to a bounded function.}
    Let
    \[
    \overline \psi=\esssup_{u\sim F_0}\psi_a(u;m,b),
    \qquad
    \underline \psi=\essinf_{u\sim F_0}\psi_a(u;m,b),
    \]
    so that $\eta_a(m,b)=\overline \psi-\underline \psi$. If $\eta_a(m,b)=0$, then $\psi_a$ is constant $F_0$-almost surely, and because $F_1\ll F_0$, the same is true $F_1$-almost surely, giving $\Delta_a(m,b)=0$. Otherwise define
    \[
    \rho(u)=\frac{\psi_a(u;m,b)-\underline \psi}{\eta_a(m,b)}.
    \]
    Then $0\le \rho(u)\le 1$ $F_0$-almost surely and, by absolute continuity, $F_1$-almost surely. The constant term cancels in \eqref{eq:delta_integral_appendix}, yielding
    \begin{equation}
        \Delta_a(m,b)=\eta_a(m,b)\int \rho(u)\{F_1(du)-F_0(du)\}.
        \label{eq:delta_eta_r}
    \end{equation}
    
    \medskip
    \noindent\textit{Step 3: likelihood-ratio total-variation bound.}
    Here total variation is $\sup_C|F_1(C)-F_0(C)|$. If $\gamma_a(m,b)=\infty$, then $\sup_C|F_1(C)-F_0(C)|\le1$, so \eqref{eq:delta_eta_r} immediately gives $|\Delta_a(m,b)|\le\eta_a(m,b)$. Hence suppose for the remainder of this step that $\gamma_a(m,b)<\infty$. For any measurable set $C$,
    \[
    F_1(C)-F_0(C)
    =
    \int_C(h-1)\,dF_0.
    \]
    The bound $h\le\gamma_a$ gives $F_1(C)-F_0(C)\le(\gamma_a-1)F_0(C)$, while the trivial bound $F_1(C)\le1$ gives $F_1(C)-F_0(C)\le1-F_0(C)$. Hence
    \[
    F_1(C)-F_0(C)
    \le
    \min\{(\gamma_a-1)F_0(C),\,1-F_0(C)\}
    \le
    \frac{\gamma_a-1}{\gamma_a}.
    \]
    Similarly,
    \[
    F_0(C)-F_1(C)=\int_C(1-h)\,dF_0\le F_0(C),
    \]
    and from $1=\int h\,dF_0\le \gamma_aF_0(C^c)+F_1(C)$ we obtain $F_0(C)-F_1(C)\le(\gamma_a-1)/\gamma_a$. Thus the total variation distance satisfies
    \[
    \sup_C |F_1(C)-F_0(C)|\le \frac{\gamma_a-1}{\gamma_a}.
    \]
    Since $0\le \rho\le1$, the Hahn decomposition characterization of signed measures gives
    \[
    \left|\int \rho(u)\{F_1(du)-F_0(du)\}\right|
    \le
    \sup_C |F_1(C)-F_0(C)|
    \le
    \frac{\gamma_a-1}{\gamma_a}.
    \]
    Combining this display with \eqref{eq:delta_eta_r} proves \eqref{eq:sharp_additive_bound} in the finite-$\gamma_a$ case; the infinite case was handled at the start of this step.
    
    \medskip
    \noindent\textit{Step 4: local endpoint attainability for finite $\gamma_a$.}
    This step does not fix any observed conditional outcome margin or any compatibility constraint across other strata. If $\gamma_a=1$, take $F_1=F_0$, which gives $\Delta_a(m,b)=0$. Suppose $1<\gamma_a<\infty$. Let $F_0$ be the uniform law on $[0,1]$, let $C^\star=[0,1/\gamma_a]$, and define
    \[
    \frac{dF_1}{dF_0}(u)=\gamma_a\mathbf{1}\{u\in C^\star\}.
    \]
    This is a probability density with essential supremum $\gamma_a$. Choose values $\underline \psi$ and $\overline \psi$ in the allowable support of the conditional mean of $Y(1,m)$ with $\overline \psi-\underline \psi=\eta_a(m,b)$; for $Y\in[L_Y,U_Y]$ this requires $\eta_a(m,b)\le U_Y-L_Y$. Set
    \[
    \psi_a(u;m,b)=\underline \psi+(\overline \psi-\underline \psi)\mathbf{1}\{u\in C^\star\}.
    \]
    Then the essential range of $\psi_a$ is $\eta_a(m,b)$ and
    \[
    \Delta_a(m,b)=(\overline \psi-\underline \psi)\frac{\gamma_a-1}{\gamma_a}
    =
    \eta_a(m,b)\frac{\gamma_a(m,b)-1}{\gamma_a(m,b)}.
    \]
    Replacing $\mathbf{1}\{u\in C^\star\}$ by $\mathbf{1}\{u\notin C^\star\}$ attains the lower endpoint. Thus both endpoints are attainable for the primitive pointwise sensitivity function under the limited finite-$\gamma_a$ conditions stated in the theorem. When $\gamma_a=\infty$, the envelope $\eta_a(m,b)$ is the limiting total-variation/range bound and finite-$\gamma_a$ attainability is not asserted.
\end{proof}

\subsection{Proof of Proposition \ref{prop:bridge_tightens}}
\label{proof:prop:bridge_tightens}
\begin{proof}
    Fix $a\in\{0,1\}$, $m\in\calM$, and $b$ in the support of $B(m)$. Let
    \[
    \mathcal X_b=\{x:B(m,x)=b\}
    \]
    denote the set of covariate values mapping to the bridge score value $b$
    at mediator $m$.
    
    \medskip
    \noindent\textit{Step 1: mixture representation of the bridge score
        conditional laws.}
    By the law of total probability,
    \begin{align}
        F(du\mid A=a,M=m,B(m)=b)
        &=
        \int_{\mathcal X_b}
        F(du\mid A=a,M=m,X=x)\,
        F(dx\mid A=a,M=m,B(m)=b),
        \label{eq:mix_num}\\
        F(du\mid A=a,B(m)=b)
        &=
        \int_{\mathcal X_b}
        F(du\mid A=a,X=x)\,
        F(dx\mid A=a,B(m)=b).
        \label{eq:mix_den}
    \end{align}
    By Lemma~\ref{lem:bridge_balancing},
    \[
    F(dx\mid A=a,M=m,B(m)=b)
    =
    F(dx\mid A=a,B(m)=b).
    \]
    Denote this common mixing measure by $\mu_{a,b}(dx)$. Then
    \begin{align}
        F(du\mid A=a,M=m,B(m)=b)
        &=
        \int_{\mathcal X_b}
        F(du\mid A=a,M=m,X=x)\,\mu_{a,b}(dx),
        \label{eq:mix_num2}\\
        F(du\mid A=a,B(m)=b)
        &=
        \int_{\mathcal X_b}
        F(du\mid A=a,X=x)\,\mu_{a,b}(dx).
        \label{eq:mix_den2}
    \end{align}
    
    \medskip
    \noindent\textit{Step 2: setwise proof of \eqref{eq:gamma_tightens}.}
    Let
    \[
    \bar\gamma_a(m,b)=\esssup_{x\in\mathcal X_b}\gamma_a^\ast(m,x),
    \]
    where the essential supremum is with respect to $\mu_{a,b}$. For every measurable set $C$ in the support of $U$, the full covariate likelihood-ratio condition implies
    \[
    F(C\mid A=a,M=m,X=x)
    \le
    \gamma_a^\ast(m,x)F(C\mid A=a,X=x)
    \le
    \bar\gamma_a(m,b)F(C\mid A=a,X=x)
    \]
    for $\mu_{a,b}$-almost every $x\in\mathcal X_b$. Integrating this setwise inequality over the common bridge score mixing law gives
    \begin{align*}
    F(C\mid A=a,M=m,B(m)=b)
    &=
    \int_{\mathcal X_b}F(C\mid A=a,M=m,X=x)\,\mu_{a,b}(dx)\\
    &\le
    \bar\gamma_a(m,b)
    \int_{\mathcal X_b}F(C\mid A=a,X=x)\,\mu_{a,b}(dx)\\
    &=
    \bar\gamma_a(m,b)F(C\mid A=a,B(m)=b).
    \end{align*}
    Thus $F_{U\mid A=a,M=m,B(m)=b}\ll F_{U\mid A=a,B(m)=b}$ and its Radon--Nikod\'ym derivative is bounded by $\bar\gamma_a(m,b)$ almost surely under the latter law. Taking the essential supremum of this derivative yields
    \[
    \gamma_a(m,b)
    \le
    \bar\gamma_a(m,b)
    =
    \esssup_{x\in\mathcal X_b}\gamma_a^\ast(m,x),
    \]
    which proves \eqref{eq:gamma_tightens}.
    
    Taking the essential supremum over $b$,
    \[
    \esssup_b\gamma_a(m,b)
    \le
    \esssup_b\esssup_{x\in\mathcal X_b}\gamma_a^\ast(m,x)
    =
    \esssup_x\gamma_a^\ast(m,x),
    \]
    because the collection $\{\mathcal X_b\}_b$ partitions the support of $X$ for fixed $m$.

    \medskip
    \noindent\textit{Step 3: comparison for the chi-square residual-selection parameter.}
    For $x\in\mathcal X_b$, write
    \[
    P_x(du)=F(du\mid A=a,M=m,X=x),
    \qquad
    Q_x(du)=F(du\mid A=a,X=x).
    \]
    By Step 1, the bridge score selected and reduced latent laws are the mixtures
    \[
    P(du)=\int_{\mathcal X_b}P_x(du)\,\mu_{a,b}(dx),
    \qquad
    Q(du)=\int_{\mathcal X_b}Q_x(du)\,\mu_{a,b}(dx).
    \]
    Consider the joint laws
    \[
    \widetilde P(dx,du)=P_x(du)\,\mu_{a,b}(dx),
    \qquad
    \widetilde Q(dx,du)=Q_x(du)\,\mu_{a,b}(dx).
    \]
    Their chi-square divergence satisfies
    \[
    D_{\chi^2}(\widetilde P\|\widetilde Q)
    =
    \int_{\mathcal X_b}
    D_{\chi^2}(P_x\|Q_x)\,\mu_{a,b}(dx)
    =
    \int_{\mathcal X_b}
    \chi_a^\ast(m,x)\,\mu_{a,b}(dx).
    \]
    Since $P$ and $Q$ are the $U$-marginals of $\widetilde P$ and $\widetilde Q$, the data-processing inequality for $f$-divergences gives
    \[
    \chi_a(m,b)
    =
    D_{\chi^2}(P\|Q)
    \le
    D_{\chi^2}(\widetilde P\|\widetilde Q).
    \]
    Therefore,
    \[
    \chi_a(m,b)
    \le
    \int_{\mathcal X_b}
    \chi_a^\ast(m,x)\,\mu_{a,b}(dx)
    =
    \E\{\chi_a^\ast(m,X)\mid A=a,B(m)=b\}
    \le
    \esssup_{x\in\mathcal X_b}\chi_a^\ast(m,x),
    \]
    which proves \eqref{eq:chi_tightens}.

    \medskip
    \noindent\textit{Step 4: outcome-relevance comparisons under
        $X\indep U\mid A=a,B(m)=b$.}
    Under the additional assumption
    \[
    X\indep U\mid A=a,B(m)=b,
    \]
    the conditional law of $X$ given $(A=a,B(m)=b,U=u)$ does not depend on
    $u$. Denote this common measure by $\mu_{a,b}(dx)$. Then, by iterated
    expectation,
    \begin{align*}
        \psi_a(u;m,b)
        &=
        \E\{Y(1,m)\mid A=a,B(m)=b,U=u\}\\
        &=
        \int_{\mathcal X_b}
        \E\{Y(1,m)\mid A=a,X=x,U=u\}\,\mu_{a,b}(dx)\\
        &=
        \int_{\mathcal X_b}
        \psi_a^\ast(u;m,x)\,\mu_{a,b}(dx).
    \end{align*}
    Let expectation and variance below be taken with respect to the common law of $U\mid A=a,B(m)=b$, which is also the law of $U\mid A=a,X=x$ for $x\in\mathcal X_b$ under the independence condition. Define
    \[
    \bar\psi_a^\ast(m,x)
    =
    \E\{\psi_a^\ast(U;m,x)\mid A=a,X=x\}.
    \]
    Since
    $\E\{\psi_a(U;m,b)\mid A=a,B(m)=b\}
    =
    \int_{\mathcal X_b}\bar\psi_a^\ast(m,x)\,\mu_{a,b}(dx)$,
    Jensen's inequality gives
    \begin{align*}
    v_a(m,b)
    &=
    \Var\!\left\{
    \int_{\mathcal X_b}
    \psi_a^\ast(U;m,x)\,\mu_{a,b}(dx)
    \,\middle|\, A=a,B(m)=b
    \right\}\\
    &=
    \E\!\left[
    \left\{
    \int_{\mathcal X_b}
    \{\psi_a^\ast(U;m,x)-\bar\psi_a^\ast(m,x)\}\,
    \mu_{a,b}(dx)
    \right\}^2
    \middle| A=a,B(m)=b
    \right]\\
    &\le
    \int_{\mathcal X_b}
    \E\!\left[
    \{\psi_a^\ast(U;m,x)-\bar\psi_a^\ast(m,x)\}^2
    \mid A=a,X=x
    \right]
    \mu_{a,b}(dx)\\
    &\le
    \esssup_{x\in\mathcal X_b}v_a^\ast(m,x),
    \end{align*}
    proving \eqref{eq:v_tightens}.

    The same representation gives the range comparison used for Theorem~\ref{thm:range_bound}. Fix $u$ and $u'$. Then
    \[
    \psi_a(u;m,b)-\psi_a(u';m,b)
    =
    \int_{\mathcal X_b}
    \{\psi_a^\ast(u;m,x)-\psi_a^\ast(u';m,x)\}\,\mu_{a,b}(dx).
    \]
    Therefore,
    \begin{align*}
        |\psi_a(u;m,b)-\psi_a(u';m,b)|
        &\le
        \int_{\mathcal X_b}
        |\psi_a^\ast(u;m,x)-\psi_a^\ast(u';m,x)|\,\mu_{a,b}(dx)\\
        &\le
        \esssup_{x\in\mathcal X_b}
        |\psi_a^\ast(u;m,x)-\psi_a^\ast(u';m,x)|\\
        &\le
        \esssup_{x\in\mathcal X_b}\eta_a^\ast(m,x).
    \end{align*}
    Taking the essential supremum over $u$ and $u'$ yields
    \[
    \eta_a(m,b)
    =
    \esssup_u \psi_a(u;m,b)-\essinf_u \psi_a(u;m,b)
    =
    \esssup_{u,u'}|\psi_a(u;m,b)-\psi_a(u';m,b)|
    \le
    \esssup_{x\in\mathcal X_b}\eta_a^\ast(m,x),
    \]
    which proves the displayed range comparison in the proposition.
\end{proof}

\subsection{Proof of Corollary \ref{cor:scalar_reduction}}
\label{proof:cor:scalar_reduction}
\begin{proof}
    The proof has three parts: (i) the scalar identification \eqref{eq:theta_scalar}; (ii) the upper bound $|\bar{\Delta}_a|\le \bar{\Xi}_a$; and (iii) the comparison $\bar{\Xi}_a\le \esssup_{(m,b)}\Xi_a(m,b)$ together with its equality condition.
    
    \smallskip
    \noindent\textit{Part (i): scalar identification.}
    By Theorem~\ref{thm:bridge_identification},
    \[
    \theta
    =
    \E\!\left[
    \int_{\calM}
    \bigl\{
    \mu_1\bigl(m,B(m)\bigr)
    -\Delta_1\bigl(m,B(m)\bigr)
    +\Delta_0\bigl(m,B(m)\bigr)
    \bigr\}
    f_0(m\mid X)\,\nu(dm)
    \right].
    \]
    The integrand is integrable under the moment conditions of Theorem~\ref{thm:bridge_identification}, so by Fubini-Tonelli the outer expectation and the inner integral may be exchanged, and linearity of expectation yields
    \[
    \theta
    =
    \theta_{\mathrm{SI}}
    +\bar{\Delta}_0-\bar{\Delta}_1,
    \]
    with $\theta_{\mathrm{SI}}$, $\bar{\Delta}_0$, $\bar{\Delta}_1$ defined in \eqref{eq:theta_SI_def} and \eqref{eq:delta_bar}. This establishes \eqref{eq:theta_scalar}. Consequently, $\theta$ depends on the pointwise bridge score sensitivity functions $(\Delta_0,\Delta_1)$ only through the arm-specific scalar functionals $(\bar{\Delta}_0,\bar{\Delta}_1)$.
    
    \smallskip
    \noindent\textit{Part (ii): upper bound and sharpness over the rectangular sensitivity class.}
    By the assumed pointwise envelope, for every $(m,b)$ in the support,
    \begin{equation}
        \bigl|\Delta_a(m,b)\bigr|\le \Xi_a(m,b),
        \qquad a\in\{0,1\}.
        \label{eq:pointwise_bound_in_cor_proof}
    \end{equation}
    Since $f_0(\cdot\mid X)\ge 0$, applying the triangle inequality for integrals and \eqref{eq:pointwise_bound_in_cor_proof} gives
    \begin{align}
        \bigl|\bar{\Delta}_a\bigr|
        &=
        \left|\E\!\left[\int_{\calM}\Delta_a\bigl(m,B(m)\bigr)\,f_0(m\mid X)\,\nu(dm)\right]\right|\\
        &\le
        \E\!\left[\int_{\calM}\bigl|\Delta_a\bigl(m,B(m)\bigr)\bigr|\,f_0(m\mid X)\,\nu(dm)\right]\\
        &\le
        \E\!\left[\int_{\calM}\Xi_a\bigl(m,B(m)\bigr)\,f_0(m\mid X)\,\nu(dm)\right]
        =
        \bar{\Xi}_a.
    \end{align}
    This establishes \eqref{eq:xi_bar}. To see sharpness over the rectangular class, fix $s_a\in\{-1,1\}$ and set $\Delta_a(m,b)=s_a\Xi_a(m,b)$ on the integration support, with any measurable extension outside that support. This function satisfies the pointwise envelope and gives $\bar\Delta_a=s_a\bar\Xi_a$. Hence both endpoints $\pm\bar\Xi_a$ are attainable for each arm. Choosing $(\Delta_0,\Delta_1)=(\Xi_0,-\Xi_1)$ attains the upper endpoint of \eqref{eq:theta_aggregate_sharp_interval}, and choosing $(\Delta_0,\Delta_1)=(-\Xi_0,\Xi_1)$ attains the lower endpoint.
    
    \smallskip
    \noindent\textit{Part (iii): comparison with the pointwise essential supremum.}
    By monotonicity of the integral,
    \[
    \bar{\Xi}_a
    =
    \E\!\left[\int_{\calM}\Xi_a\bigl(m,B(m)\bigr)\,f_0(m\mid X)\,\nu(dm)\right]
    \le
    \esssup_{(m,b)}\Xi_a(m,b),
    \]
    with equality if and only if $\Xi_a(m,B(m))$ equals $\esssup_{(m,b)}\Xi_a(m,b)$, $f_0$-almost surely on the support of $(M,B(m))$ induced by $f_0(\cdot\mid X)$.
\end{proof}

\end{document}